\documentclass{article}
\usepackage{ijcai26}

\usepackage{times}
\usepackage{soul}
\usepackage{url}
\usepackage[hidelinks]{hyperref}
\usepackage[utf8]{inputenc}
\usepackage[small]{caption}
\usepackage{graphicx}
\usepackage{amsmath}
\usepackage{amsthm}
\usepackage{amsfonts}
\usepackage{amssymb}
\usepackage{booktabs}
\usepackage{algorithm}
\usepackage{algorithmic}
\usepackage[switch]{lineno}
\usepackage{mathtools}
\usepackage{subcaption}

\theoremstyle{definition}
\newtheorem{dfn}{Definition}[section]
\newtheorem{ass}[dfn]{Assumption}
\newtheorem{prop}[dfn]{Proposition}

\newtheorem{thm}[dfn]{Theorem}

\title{Off-Policy Evaluation and Learning for Survival Outcomes Under Censoring}

\author{
Kohsuke Kubota$^1$
\and
Mitsuhiro Takahashi$^1$\and
Yuta Saito$^{2}$\\
\affiliations
$^1$NTT DOCOMO, INC.\\
$^2$Hanjuku-kaso, Co, Ltd.\\
\emails
\{kousuke.kubota.xt, mitsuhiro.takahashi.zp\}@nttdocomo.com,
saito@hanjuku-kaso.com
}

\begin{document}
\maketitle

\begin{abstract}
Optimizing survival outcomes, such as patient survival or customer retention, is a critical objective in data-driven decision-making. Off-Policy Evaluation~(OPE) provides a powerful framework for assessing such decision-making policies using logged data alone, without the need for costly or risky online experiments in high-stakes applications.
However, typical estimators are not designed to handle right-censored survival outcomes, as they ignore unobserved survival times beyond the censoring time, leading to systematic underestimation of the true policy performance.
To address this issue, we propose a novel framework for OPE and Off-Policy Learning~(OPL) tailored for survival outcomes under censoring.
Specifically, we introduce IPCW-IPS and IPCW-DR, which employ the Inverse Probability of Censoring Weighting technique to explicitly deal with censoring bias.
We theoretically establish that our estimators are unbiased and that IPCW-DR achieves double robustness, ensuring consistency if either the propensity score or the outcome model is correct.
Furthermore, we extend this framework to constrained OPL to optimize policy value under budget constraints.
We demonstrate the effectiveness of our proposed methods through simulation studies and illustrate their practical impacts using public real-world data for both evaluation and learning tasks.
\end{abstract}

\section{Introduction}
Evaluating and implementing policies~(treatment strategies) aimed at maximizing survival outcomes presents a critical challenge in decision-making in diverse real-world domains, ranging from extending survival in healthcare~\cite{geng2015optimal,murphy2003optimal} to reducing customer churn in marketing~\cite{ascarza2018pursuit}.
The core focus of these applications is not limited to evaluating short-term effects, but encompasses a broader understanding of a policy's impact on the long-term survival trajectory.

Off-policy evaluation~(OPE) provides a powerful framework to estimate the performance of new decision-making policies using only offline logged data, which is crucial in domains where online experiments are costly or impractical~\cite{wang2017optimal}.
Conventional approaches, such as Inverse Propensity Scoring~(IPS)~\cite{horvitz1952generalization,rosenbaum1983central} and Doubly Robust~(DR)~\cite{dudik2011doubly,robins1995semiparametric}, have been widely adopted for evaluating policies using offline logged data collected under different (logging) policies.

However, despite their broad applicability, \textbf{these typical estimators are not designed to handle right-censored survival outcomes}.
This is because these estimators are designed primarily for immediate rewards or outcomes defined over a fixed time horizon.
They also implicitly assume that all outcomes are fully observed, which is directly violated by right-censoring, a defining characteristic of survival data.
A naive application of these estimators cannot distinguish censoring from true event occurrence, and therefore wrongly treats observed times as final event times.
This induces systematic underestimation bias because typical methods fail to account for potential survival beyond the censoring point.

To address this issue, we propose a novel framework for off-policy evaluation and learning under practical constraints designed for censored survival outcomes.
Specifically, we develop new statistical estimators, namely \textbf{Inverse Probability of Censoring Weighted-Inverse Propensity Scoring (IPCW-IPS)} and \textbf{IPCW-Doubly Robust (IPCW-DR)}, which explicitly consider censoring induced bias.
IPCW-IPS additionally re-weights logged data via estimated censoring probabilities beyond the typical importance weighting regarding policy probabilities. 
IPCW-DR further incorporates outcome modeling, which enjoys the doubly robust property with respect to the propensity score and the outcome model. 
Consequently, it ensures reliability in practical applications and yields lower variance compared to IPCW-IPS.
Furthermore, we extend our framework to the problem of constrained Off-Policy Learning~(OPL), enabling policy optimization for survival outcomes subject to practical budget constraints.

Finally, we empirically demonstrate the effectiveness of our proposed framework. The simulation results demonstrate that our proposed estimators, particularly IPCW-DR, achieve superior accuracy than standard OPE estimators. Moreover, the application to the Starbucks Customer Reward Program dataset highlights the practical effectiveness of our approach in complex scenarios. Specifically, we demonstrate how our framework can better optimize promotional offers to minimize the customer purchase cycle under budget constraints.

Our key contributions can be summarized as follows.
\begin{itemize}
    \item We develop censoring aware estimators for off-policy evaluation of survival outcomes leveraging the technique of inverse probability of censoring weighting and provide their theoretical advantages.
    \item We introduce an iterative algorithm for off-policy learning with practical constraints that targets optimization of survival outcomes.
    \item We demonstrate the effectiveness of the proposed methods in terms of both evaluation and learning tasks, through simulation studies and a real world application.
\end{itemize}

Due to space limitations, we provide a comprehensive review of related work in Appendix~\ref{app:related_work}.

\section{Preliminaries}

\subsection{Problem Formulation}
Here we consider the OPE problem for an evaluation policy $\pi_e$ under right-censoring, using data collected by a logging policy $\pi_0$.
Let $x \in \mathcal{X}$ be the context vector representing individual characteristics, $a \in \mathcal{A}$ be the discrete action selected according to a policy $\pi(a \mid x)$. 
The logged dataset $\mathcal{D} = \{ ( x_i, a_i, T_i, r_i ) \}_{i=1}^{n}$ is collected under a logging policy $\pi_0$, where the underlying random variables are generated as
\begin{equation}
    (x_i, a_i, L_i, C_i) \sim p(x) \pi_0(a_i \mid x_i) p(L_i, C_i \mid x_i, a_i).
\end{equation}
Here, $L_i$ and $C_i$ denote the true (latent) survival time and the censoring time, respectively.
The observed data consist only of the observed time $T_i = \min\{L_i, C_i\}$ and an event indicator $r_i = \mathbb{I} \{ L_i \leq C_i \}$.
Thus, for censored individuals ($r_i=0$), we only know that their true survival time exceeds the observed time $T_i$.
Our goal is to accurately estimate the performance (policy value) of an evaluation policy $\pi_e$ regarding the true survival time using only the logged data $\mathcal{D}$.

\paragraph{Policy Value Definition.}
A key quantity of interest in our work is the expected survival probability at an arbitrary time $t$.
Let $S(x, a, t) \coloneq P(L > t \mid x, a)$ be the survival function given context $x$ and action $a$.
The expected survival probability under policy $\pi_e$ at time $t$, as a type of policy value, is defined as
\begin{equation}
V(\pi_e, t) \coloneq  \mathbb{E}_{p(x) \pi_e(a \mid x)} [S(x, a, t)].
\end{equation}

Another important quantity of interest is the Restricted Mean Survival Time~(RMST)~\cite{uno2014moving}, which is the expected survival time up to a pre-specified finite time horizon $\tau$. This version of the policy value is defined as follows.
\begin{equation}
\label{eq:def-rmst}
V^{\tau}(\pi_e) \coloneq \int_{0}^{\tau} V(\pi_e, t) dt.
\end{equation}
Note that $\tau$ is not a tunable hyperparameter but a domain-specific objective that defines the evaluation metric itself.

\paragraph{Assumptions.}
We begin by outlining some regularity assumptions necessary for analyzing the theoretical properties of OPE estimators under censoring.
First, the following common support assumption, standard in the OPE literature~\cite{sachdeva2020off}, ensures that the evaluation policy does not place a positive probability on actions that the logging policy never takes.
It is formally stated as follows.
\begin{ass}[Common Support]
\label{ass:common_support}
If $\pi_e (a \mid x) > 0$, then $\pi_0(a \mid x) > 0$ for all $x \in \mathcal{X}$ and $a \in \mathcal{A}$.
\end{ass}
We also introduce an assumption regarding censoring.
\begin{ass}[Conditional Independent Censoring]
\label{ass:independent_censoring}
Given the context $x$ and action $a$, the true survival time $L$ and the censoring time $C$ are conditionally independent
\begin{equation*}
L \perp C \mid (x, a).
\end{equation*}
\end{ass}
This assumption serves as the theoretical basis for conventional methods in survival analysis, such as the Kaplan-Meier estimator and the Cox proportional hazards model~\cite{fleming2013counting}.

\begin{figure}
\centering
\includegraphics[width=0.9\columnwidth]{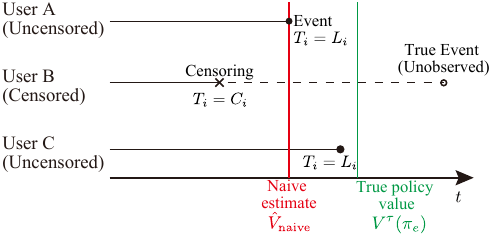}
\caption{
Illustration of the underestimation bias inherent in naive OPE. 
While the true event time for User B is unobserved due to censoring, naive estimators treat the censoring time as the event time, leading to a systematic underestimation of the survival time.}
\label{fig:teaser}
\end{figure}

\subsection{Naive OPE Under Censoring}
A naive approach to addressing this estimation problem is to directly apply conventional estimators in OPE, such as IPS or DR, to the logged dataset ignoring the censoring issue~\cite{saito2021counterfactual}.
These standard estimators are typically designed for fully observed outcomes and, when naively applied to survival data, they treat the observed time $T_i$ as the true event time regardless of the censoring indicator $r_i$.
Specifically, by treating the observed indicator $\mathbb{I} \{ T_i > t \}$ as the reward, the naive IPS estimator is given by
 $$\hat{V}_{\text{IPS}} (\pi_e, t; \mathcal{D}) \coloneq \frac{1}{n} \sum_{i=1}^{n} \frac{\pi_e(a_i \mid x_i)}{\pi_0(a_i \mid x_i)} \mathbb{I} \{ T_i > t\}.$$
We can also define the naive DR estimator by combining IPS with an outcome model $\hat{S}(x, a, t)$ trained on the observed times.\footnote{We formally define naive DR in Eq.~(\ref{eq:DR}) in Appendix~\ref{app:proof}.}
Since these estimators are not designed to account for the censoring mechanism, they suffer from systematic underestimation bias, explicitly characterized as follows.

\begin{prop}
\label{prop:bias_naive_estimators}
Under the presence of censoring~(that is, $P(L > C) > 0$) and Assumptions~\ref{ass:common_support} and \ref{ass:independent_censoring}, the naive IPS and DR estimators are biased as
\begin{align*}
&\text{Bias} [\hat{V}_{\text{IPS}}] = \text{Bias} [\hat{V}_{\text{DR}}] \\ 
&\coloneq \mathbb{E}_{p(x) \pi_e(a \mid x)} [ S(x, a, t) (G(t \mid x, a) - 1) ],
\end{align*}
where $G(t \mid x, a) \coloneq P(C > t \mid x, a)$ is the censoring survival function.
\end{prop}

The proof is provided in Appendix~\ref{app:bias_naive_estimators}.
Since $G(t \mid x, a) \leq 1$, the bias term is always non-positive, meaning that these naive estimators systematically underestimate the policy value, as visually depicted in Figure~\ref{fig:teaser}.
We can also see that the bias increases under severe censoring conditions, such as short follow-up periods or low event rates.
In such cases, $G(t \mid x, a)$ becomes significantly smaller than one, resulting in a more severe underestimation bias.

\section{Proposed Approach: Censoring-Aware Off-Policy Evaluation and Learning}
To obtain unbiased estimates of expected survival probability or RMST, it is essential to incorporate the censoring mechanism directly into the OPE framework. This section proposes two new estimators that properly account for censoring and yield unbiased or at least bias-reduced policy value estimates.

\subsection{Unbiased Estimation of Survival Policy Value}
\label{subsec:proposed-ope}
Our first proposed estimator, namely \textbf{Inverse Probability of Censoring Weighted-Inverse Propensity Scoring~(IPCW-IPS)}, directly corrects for the censoring bias.
As shown in the derivation of Proposition~\ref{prop:bias_naive_estimators}~(see Appendix~\ref{app:bias_naive_estimators}), the expectation of the naive outcome $\mathbb{I}\{T_i > t\}$ yields the confounded quantity $S(x_i, a_i, t) G(t \mid x_i, a_i)$ rather than the true survival probability $S(x_i, a_i, t)$.
To deal with this bias, we propose to employ \textbf{the Inverse Probability of Censoring Weighting~(IPCW)} technique~\cite{robins1995analysis,robins2000correcting}.
Specifically, we define IPCW-IPS by additionally weighting each observed indicator $\mathbb{I} \{ T_i > t \}$ with the inverse of the estimated probability that the sample remains uncensored at time $t$, denoted by $\hat{G}(t \mid x_i, a_i)$ as
\begin{equation}
    \hat{V}_{\text{IPCW-IPS}}(\pi_e, t; \mathcal{D}) \coloneq \frac{1}{n} \sum_{i=1}^{n} w(x_i, a_i) \frac{\mathbb{I} \{ T_i > t \} }{\hat{G}(t \mid x_i, a_i)}.
\end{equation}
where $w(x_i, a_i) \coloneq \pi_e(a_i \mid x_i) / \pi_0(a_i \mid x_i)$. 

While IPCW-IPS is designed to be unbiased, it may be susceptible to high variance particularly when censoring is heavy~(i.e., when $\hat{G}(t \mid x, a)$ is small), as the inverse probability weight becomes large.
To mitigate this variance inflation and improve estimation reliability further, we extend IPCW-IPS to \textbf{IPCW-Doubly Robust~(IPCW-DR)} defined as
\begin{align}
\begin{split}
& \hat{V}_{\text{IPCW-DR}} (\pi_e, t; \mathcal{D}) \\
& \coloneq \frac{1}{n} \sum_{i=1}^{n} \Biggl( w(x_i, a_i) \left( \frac{\mathbb{I} \{ T_i > t \} }{\hat{G}(t \mid x_i, a_i)} - \hat{S}(x_i, a_i, t) \right) \\
&\hspace{4cm}+ \mathbb{E}_{\pi_e(a \mid x)} [\hat{S}(x_i, a, t)] \Biggr).
\end{split}
\end{align}
The IPCW-DR estimator extends the standard DR framework to the setting of censored survival outcomes.
Structurally, this estimator combines a model-based estimate derived via the direct method~(DM)~\cite{beygelzimer2009offset}, with a bias-correction term based on the weighted residual between the IPCW-adjusted reward and the outcome model.

We first establish the unbiasedness of our proposed estimators. The following theorem guarantees that both IPCW-IPS and IPCW-DR provide unbiased estimates of the policy value, with IPCW-DR further enjoying double robustness.
\begin{thm}
\label{thm:unbiasedness}
Under Assumptions~\ref{ass:common_support} and \ref{ass:independent_censoring}, provided that the censoring model $G(t \mid x, a)$ is correctly specified
\begin{enumerate}
    \item The IPCW-IPS estimator is unbiased for the expected survival probability $V(\pi_e, t)$ if the propensity score $\pi_0(a \mid x)$ is correctly specified.
    \item The IPCW-DR estimator is unbiased for $V(\pi_e, t)$ if either the propensity score $\pi_0(a \mid x)$ or the outcome model $\hat{S}(x, a, t)$ is correctly specified.
\end{enumerate}
The proofs are provided in Appendix~\ref{app:unbiasedness}.
\end{thm}

While unbiasedness is a desirable statistical property for OPE, the variance of an estimator is an equally critical component of the MSE.
The following theorem explicitly characterizes the variance of our estimators and theoretically establishes the variance reduction property of IPCW-DR.
\begin{thm}
\label{thm:variance}
Under Assumptions~\ref{ass:common_support} and \ref{ass:independent_censoring}, provided that the propensity score $\pi_0(a \mid x)$, the censoring model $G(t \mid x, a)$, and the outcome model $S(x, a, t)$ are correctly specified, the variance of the IPCW-IPS estimator is
\begin{equation*}
\begin{split}
n\mathbb{V}[\hat{V}_{\text{IPCW-IPS}}] &= \mathbb{E}_{p(x) \pi_0(a \mid x)}\left[(w(x,a))^2 \sigma^2(x,a,t)\right] \\
&\quad + \mathbb{E}_{p(x)}[\mathbb{V}_{\pi_0(a \mid x)}[w(x,a)S(x,a,t)]] \\
&\quad + \mathbb{V}_{p(x)}[S^{\pi_e}(x,t)],
\end{split}
\end{equation*}
where $\sigma^2(x,a,t) \coloneq S(x,a,t)/G(t|x,a) - S(x,a,t)^2$ represents the variance inflation due to censoring.
Furthermore, the variance of the IPCW-DR estimator is
\begin{equation*}
\begin{split}
n\mathbb{V}[\hat{V}_{\text{IPCW-DR}}] &= \mathbb{E}_{p(x) \pi_0(a \mid x)}\left[w(x,a)^2 \sigma^2(x,a,t)\right] \\ 
&\quad + \mathbb{V}_{p(x)}[S^{\pi_e}(x,t)].
\end{split}
\end{equation*}
Consequently, IPCW-DR achieves a variance reduction of $\mathbb{E}_{p(x)}[\mathbb{V}_{\pi_0(a \mid x)}[w(x,a)S(x,a,t)]] \ge 0$, ensuring $\mathbb{V}[\hat{V}_{\text{IPCW-DR}}] \le \mathbb{V}[\hat{V}_{\text{IPCW-IPS}}]$.
The proofs are provided in Appendix~\ref{app:variance}.
\end{thm}

This theorem provides a formal justification for the instability of IPCW-IPS under heavy censoring.
Specifically, the term $\sigma^2(x,a,t)$ reveals that the variance is inversely proportional to the censoring survival function $G(t \mid x, a)$.
Thus, when censoring is heavy (i.e., $G(t \mid x, a)$ is small), IPCW-IPS can suffer from significant variance.
In contrast, IPCW-DR consistently reduces variance by leveraging the outcome model, making it a more reliable choice in practice.

While we have so far focused on estimating the point-wise survival probability $V(\pi_e, t)$, our framework naturally extends to the estimation of the RMST $V^{\tau}(\pi_e)$ defined in Eq.~(\ref{eq:def-rmst}).
Since the RMST is the integral of the survival probability, we can construct unbiased estimators for RMST by integrating IPCW-IPS and IPCW-DR over the time horizon $[0, \tau]$.
The detailed formulations and theoretical guarantees for these RMST estimators are provided in Appendix~\ref{app:extension_RMST}.

\subsection{Off-Policy Learning for Survival Maximization}
\label{subsec:policy_learning}
In addition to the problem of OPE, we can formulate the problem of learning an optimal policy to maximize the survival outcome under censoring.
Let $\Pi$ be a predefined policy class such as neural networks.
The goal of off-policy learning is to find a policy $\pi_{\theta} \in \Pi$ that maximizes the policy value based on the logged data $\mathcal{D}$.
Due to space limitations, we will focus on the policy learning objective for the expected survival probability $V(\pi, t)$.\footnote{Note that the formulation for maximizing RMST proceeds analogously by integrating the objective over the time horizon.}
Using our proposed estimators, we can define the policy learning objectives as $\theta \coloneq \arg\max_{\theta} \hat{V} (\pi_{\theta}, t; \mathcal{D})$. 

A typical approach to solve this optimization problem is to use gradient-based approach~\cite{huang2020importance,saito2024potec}, which updates the policy parameter via iterative gradient ascent as $\theta_{k + 1} \leftarrow \theta_{k} + \nabla_{\theta} V(\pi_{\theta}, t)$.
Because the true gradient 
\begin{equation}
\label{eq:policy_gradient}
\nabla_{\theta}V(\pi_{\theta}, t) =\mathbb{E}_{p(x) \pi_{\theta}(a \mid x)} \left[ S(x, a, t) \nabla_{\theta} \log \pi_{\theta} (a \mid x) \right].
\end{equation}
is unknown, we need to estimate it from the logged dataset $\mathcal{D}$ where we can apply our proposed estimators. For example, IPCW-IPS for the policy gradient can be constructed by replacing the true survival probability $S(x, a, t)$ in Eq.~\eqref{eq:policy_gradient} with the IPCW-corrected outcome $\mathbb{I} \{ T > t \} / \hat{G}(t \mid x, a)$ as
\begin{align}
\label{eq:ipcw-ips-opl}
\begin{split}
&\widehat{\nabla_{\theta} V}_{\text{IPCW-IPS}} (\pi_{\theta}, t; \mathcal{D}) \\
&\coloneq \frac{1}{n} \sum_{i=1}^{n} w^{\prime}(x, a) \frac{\mathbb{I} \{ T_i > t \} }{\hat{G}(t \mid x_i, a_i)} \nabla_{\theta} \log \pi_{\theta} (a_i \mid x_i),
\end{split}
\end{align}
where $w^{\prime}(x, a) \coloneq \pi_{\theta}(a \mid x) / \pi_0(a \mid x)$.
Similarly, the IPCW-DR estimator for the policy gradient can be constructed by replacing the true survival probability $S(x, a, t)$ in Eq.~\eqref{eq:policy_gradient} with the IPCW-DR outcome model as follows.
\begin{align}
\label{eq:ipcw-dr-opl}
\begin{split}
&\widehat{\nabla_{\theta} V}_{\text{IPCW-DR}} (\pi_{\theta}, t; \mathcal{D}) \\
&\coloneq \frac{1}{n} \sum_{i=1}^{n} \Biggl( w^{\prime}(x_i, a_i) \hat{\delta}(x_i, a_i, t) \nabla_{\theta} \log \pi_{\theta}(a_i \mid x_i) \\
&\quad+ \mathbb{E}_{\pi_{\theta}(a \mid x)} [\hat{S}(x_i, a, t) \nabla_{\theta} \log \pi_{\theta}(a_i \mid x_i)] \Biggr),
\end{split}
\end{align}
where let $\hat{\delta}(x, a, t) \coloneq \frac{\mathbb{I} \{ T > t \} }{\hat{G}(t \mid x, a)} - \hat{S}(x, a, t)$.

These estimators are unbiased against the true policy gradient $\nabla_{\theta} V(\pi_{\theta}, t)$ under the same conditions as Theorem~\ref{thm:unbiasedness}.
IPCW-DR for the policy gradient is superior to IPCW-IPS in terms of variance under the same conditions as in Section~\ref{subsec:proposed-ope}, leading to a more efficient policy learning.

\paragraph{Extension to Constrained Off-Policy Learning.}
We can further extend our OPL framework to optimize survival outcomes \textit{under a budget constraint}, a common requirement in many real-world applications~\cite{cai2023marketing,zhang2021bcorle,zhou2023direct}.

Let $c(x, a)$ be the observed cost for a given $(x,a)$. We aim to find $\pi_\theta$ that maximizes $V(\pi_\theta, t)$ subject to $C(\pi_{\theta}) \coloneq \mathbb{E}_{p(x) \pi_\theta(a \mid x)} [c(x, a)] \leq B$.
This constrained optimization problem can be transformed using Lagrangian duality~\cite{geoffrion1974lagrangean} into a min-max problem as
$\min_{\lambda \geq 0} \max_{\theta} L(\pi_\theta, t, \lambda)$, where $L(\pi_\theta, t, \lambda) = V(\pi_\theta, t) - \lambda( C(\pi_\theta) - B )$.
To solve this problem, we can employ an iterative algorithm that alternates between updating $\theta$~(gradient ascent on $L$) and the Lagrangian multiplier $\lambda$.
The true gradient of the Lagrangian is 
\begin{equation}
\nabla_{\theta} L(\pi_\theta, t, \lambda) = \mathbb{E}_{p(x) \pi_\theta(a \mid x)} [ R_{\lambda}(x, a, t) \nabla_{\theta} \log \pi_\theta(a \mid x) ],
\end{equation}
where $R_{\lambda}(x, a, t) = S(x, a, t) - \lambda c(x, a)$.
Since $c(x, a)$ is fully observed and independent of censoring, we apply the IPCW correction exclusively to the survival component $S(x, a, t)$.
Thus, we can naturally obtain the IPCW-IPS and IPCW-DR estimators for $\nabla_\theta L$ by simply replacing the reward term in Eqs.~\eqref{eq:ipcw-ips-opl} and \eqref{eq:ipcw-dr-opl} with the Lagrangian reward $ \left( \mathbb{I} \{ T_i > t \} / \hat{G}(t \mid x, a) - \lambda c(x_i, a_i ) \right)$ and its  model-augmented counterpart, respectively.
The unbiasedness of these estimators against the true Lagrangian policy gradient is proved in Appendix~\ref{app:estimators-against-Lagrangian}.
This approach enables stable offline policy learning that rigorously satisfies budget constraints while handling the censoring issue.

\begin{figure*}[t]
  \centering
  \begin{subfigure}{0.9\textwidth}
    \centering
    \includegraphics[width=0.9\linewidth]{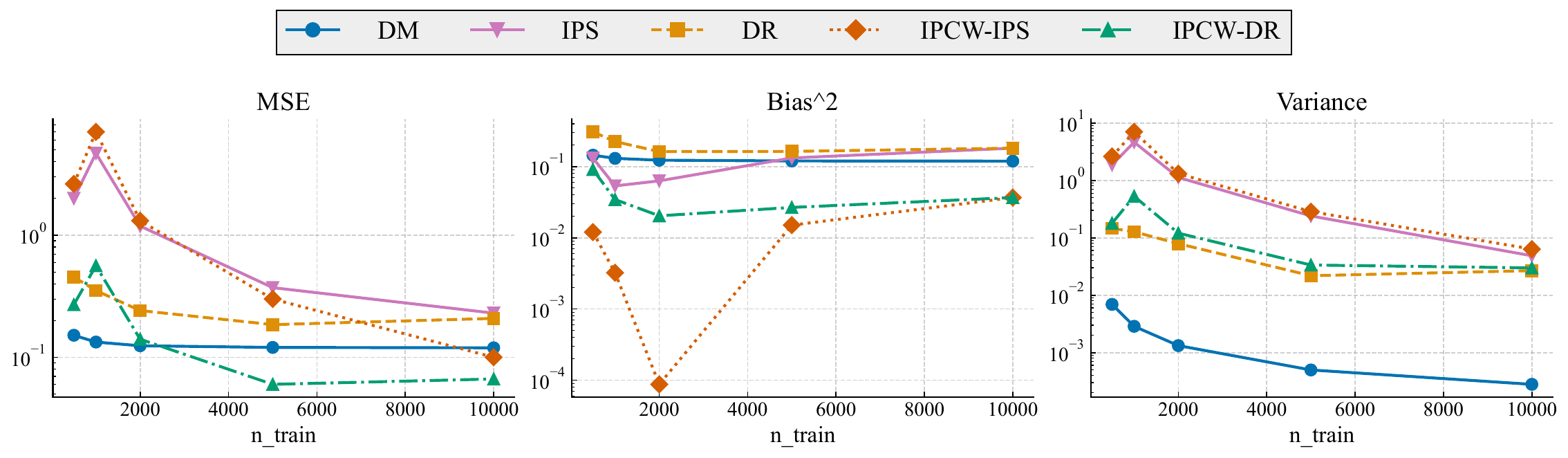} \vspace{-2mm}
    \caption{Impact of logged data size $n$ on OPE performance.}
    \label{fig:ope-n-train}
  \end{subfigure}

  \begin{subfigure}{0.9\textwidth}
    \centering
    \includegraphics[width=0.9\linewidth]{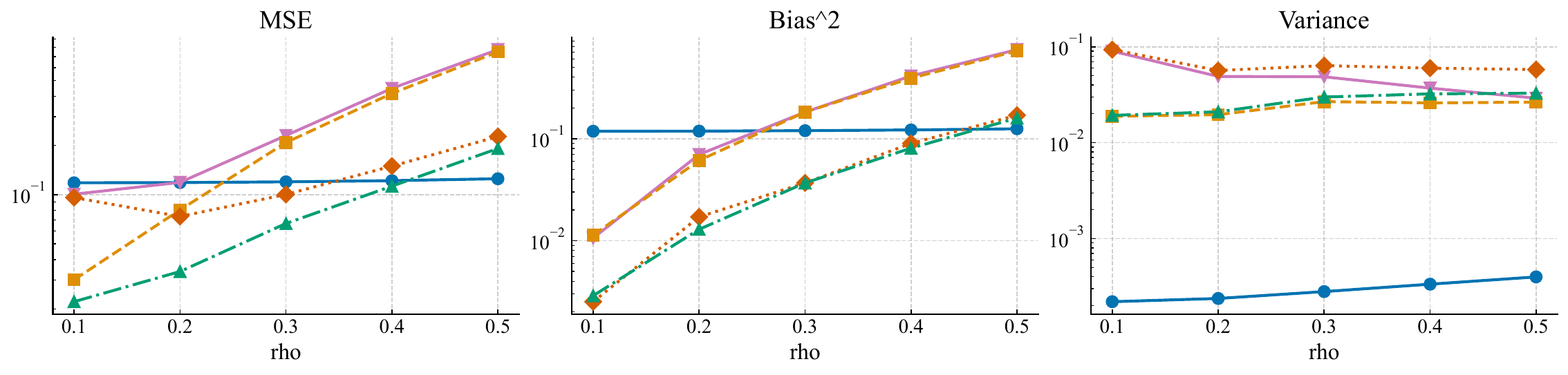} \vspace{-2mm}
    \caption{Impact of censoring rate $\rho_1$ on OPE performance.}
    \label{fig:ope-rho}
  \end{subfigure}

  \begin{subfigure}{0.9\textwidth}
    \centering
    \includegraphics[width=0.9\linewidth]{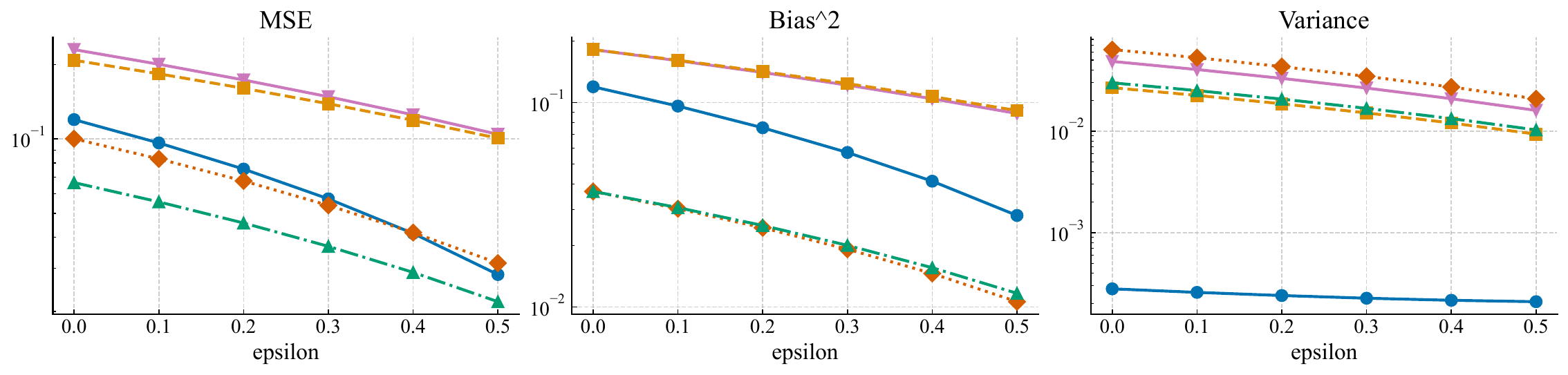} \vspace{-2mm}
    \caption{Impact of policy divergence $\epsilon$ on OPE performance.}
    \label{fig:ope-epsilon}
  \end{subfigure}
  \caption{
    OPE performance comparison across key experimental factors.
    The figure illustrates how MSE changes as we vary (a) the logged data size $n$, 
    (b) the censoring rate $\rho_1$, and (c) the divergence between the logging and evaluation policies $\epsilon$.
  }
  \label{fig:ope-results-all}
\end{figure*}

\section{Simulation Studies}
This section demonstrates the effectiveness of our proposed methods in controlled simulated environments, which allow us to systematically analyze their behavior under known and carefully varied experimental conditions.

\paragraph{Data Generating Process.}
We construct the logged dataset $\mathcal{D} = \{ (x_i, a_i, T_i, r_i)\}_{i = 1}^{n}$ to simulate OPE/L for survival outcomes under right-censoring.
The data-generating process samples context vectors $x$, selects actions $a$ using a logging policy $\pi_0$, and generates latent survival times $L$ along with censoring times $C$. 
This design enables direct control over the degree of censoring. 
Detailed descriptions of the data-generating process are provided in Appendix~\ref{app:simulation_data_generating_process}.

\paragraph{OPE Performance Comparison.}
In the OPE experiments, we compare IPCW-IPS and IPCW-DR with the naive baselines including DM, IPS and DR, all of which ignore the censoring mechanism.
We investigate the OPE accuracy under varying experimental conditions defined by three key factors: (i) training sample size $n$, (ii) censoring rate $\rho_1$, and (iii) policy divergence $\epsilon$.
First, we vary $n \in \{ 500, 1,000, 2,000, 5,000, 10,000 \}$ to evaluate asymptotic behavior of the estimators.
Second, we vary the censoring rate $\rho_1 = P(L > C)$ from $\{ 0.1, 0.2, 0.3, 0.4, 0.5 \}$ to test robustness against heavy censoring.
We obtain each target $\rho_1$ by numerically adjusting the censoring distribution parameters, as described in Appendix~\ref{app:simulation_data_generating_process}.
Third, to analyze the impact of the discrepancy between the logging and evaluation policies, we define the evaluation policy $\pi_e$ using the $\epsilon$-greedy rule: 
\begin{equation*}
\pi_e(a \mid x; \epsilon) = (1 - \epsilon) \cdot \mathbb{I} \{ a = \text{argmax}_{a'} V^\tau(x, a') \} + \frac{\epsilon}{|\mathcal{A}|}.
\end{equation*}
We configure the logging and evaluation policies in accordance with the synthetic experimental protocols in prior relevant studies~\cite{kiyohara2024off,peng2023offline}.
By varying $\epsilon \in \{ 0.1, 0.2, 0.3, 0.4, 0.5 \}$, we examine how the divergence between the logging policy $\pi_0$ and the target policy $\pi_e$ affects estimation variance.

We evaluate and compare the MSE, squared bias, and variance of the estimators over 100 independent trials.
In practice, constructing our estimators requires estimating nuisance components such as the censoring model $\hat{G}(t \mid x, a)$.
To ensure computational stability under heavy censoring, we estimate this using Cox proportional hazards models with $L_2$ regularization.
Additional implementation details are provided in Appendix~\ref{app:simulation_ope_details} and our code to reproduce the results is available as ancillary files of this arXiv submission.

\paragraph{OPL Performance Comparison.}
We evaluate the performance of OPL methods in learning a policy $\pi_\theta$, parameterized by a Multi-Layer Perceptron (MLP)~\cite{lecun2015deep}, to maximize the expected RMST. Our comparison includes two learning paradigms, including a regression-based approach~\cite{jeunen2021pessimistic}, which selects actions using RMST predictions from an outcome model, and policy gradient–based approaches~\cite{huang2020importance}, which optimize differentiable OPE estimators directly as objective functions.
We analyze the learning performance under varying training sample sizes $n$ and censoring rates $\rho_1$.
We evaluate the learned policies on an independent large-scale test set using the improvement ratio over the logging policy, defined as $V^{\tau}(\hat{\pi}) / V^{\tau}(\pi_0)$.  Additional implementation details are provided in Appendix~\ref{app:simulation_opl_details}.

\paragraph{Simulation Results.}
Figure~\ref{fig:ope-results-all} illustrates the OPE performance comparison results across varying logged data sizes $n$, censoring rates $\rho_1$, and policy divergences $\epsilon$.

Figure~\ref{fig:ope-n-train} interestingly shows that naive IPS and DR do not reduce squared bias even when we increase the sample size, and their MSE remains high. 
This observation directly reflects Proposition~\ref{prop:bias_naive_estimators}, which explains how censoring induces systematic underestimation for those typically unbiased estimators. 
In contrast, IPCW-IPS and IPCW-DR steadily lower both MSE and squared bias as $n$ grows, even in the existence of severe censoring, and these trends align with the unbiasedness guarantees in Theorem~\ref{thm:unbiasedness}. 
DM achieves the smallest variance, but it retains a non-vanishing and substantial bias because the outcome model alone cannot fully capture the underlying survival dynamics. IPCW-DR corrects this bias while maintaining a strong bias--variance balance.

Next, in Figure~\ref{fig:ope-rho}, we observe that the performance of naive estimators deteriorates rapidly as the censoring rate increases.
Our IPCW-based estimators, however, continue to produce much lower MSE even under heavy censoring. 
IPCW-DR, in particular, avoids the variance inflation that often appears in high-censoring settings, which highlights its robustness. 
Moreover, IPCW-DR remains highly stable under high policy divergence as shown in Figure~\ref{fig:ope-epsilon}, whereas naive estimators continue to show unreliable accuracy.

\begin{figure}[!t]
  \centering
  \begin{subfigure}{0.95\linewidth}
    \centering
    \includegraphics[width=\linewidth]{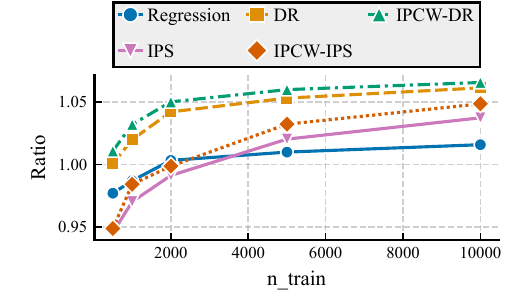} \vspace{-2mm}
    \caption{Impact of training sample size $n$ on OPL performance.} 
    \label{fig:opl-n-train}
  \end{subfigure} 
  \begin{subfigure}{0.95\linewidth}
    \centering
    \includegraphics[width=\linewidth]{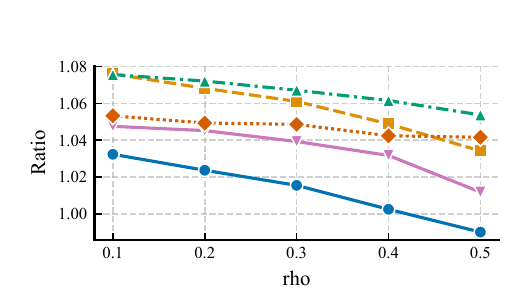} \vspace{-2mm}
    \caption{Impact of censoring rate $\rho_1$ on OPL performance.}
    \label{fig:opl-rho}
  \end{subfigure}

  \caption{
    OPL performance comparison. 
    The improvement ratio $V^\tau(\hat{\pi}) / V^\tau(\pi_0)$ quantifies how effectively each method learns a superior policy compared to the logging policy optimalities.
  }
  \label{fig:opl-all}
\end{figure}

Figure~\ref{fig:opl-all} summarizes the OPL results. In Figure~\ref{fig:opl-n-train}, policies learned based on IPCW-IPS and IPCW-DR achieve substantially higher improvement ratios compared to policy learners based on naive IPS and DR. 
Naive methods frequently underestimate survival times and thus fail to choose optimal actions, while IPCW-based learners appropriately deal with this bias and consistently identify better policies. 
Our methods continue to outperform naive learners as censoring increases as well (Figure~\ref{fig:opl-rho}). 
These results demonstrate that our IPCW-based approaches achieve robust and reliable policy improvements across a range of realistic conditions.

\section{Real-World Case Study: Starbucks Customer Retention}
\label{sec:application-real-world-dataset}
This section explores the practical effectiveness of our proposed estimators through a semi-synthetic experiment based on the Starbucks Customer Reward Program dataset\footnote{\url{https://www.kaggle.com/datasets/blacktile/starbucks-app-customer-reward-program-data}}.  
Real-world datasets capture rich behavioral patterns, but they rarely offer counterfactual outcomes or complex censoring structures that allow a rigorous comparison of off-policy estimators or learners. 
To address these limitations and to perform yet a meaningful experiment, we first construct a data-driven “Oracle’’ environment directly from the Starbucks mobile app logs. 
The dataset contains detailed customer information such as age, income, membership tenure, and gender, together with offer attributes such as difficulty, reward amount, delivery channel, and duration.  
We define a discrete action space $\mathcal{A}$ containing ten offer types.  
In this environment, the survival time $T$ measures how long a user takes to make the first purchase after receiving an offer.  
Our goal in this environment is to evaluate policies that shorten this purchase cycle, which corresponds to minimizing RMST over the validity period.  
We treat a purchase within the period as an event and regard all remaining observations as censored.

Furthermore, to examine estimator robustness under challenging censoring mechanisms, we introduce artificial censoring that depends on user demographics as their features.  
The original dataset mainly reflects simple Type I censoring due to offer expiration, which does not adequately stress-test modern OPE estimators.  
By generating censoring times that interact with demographic features, we create scenarios where the random censoring assumption fails in a controlled and interpretable way.  
This setting allows our experiment to reveal when naive estimators fail and how IPCW-based corrections address those failures.  
We describe the full data construction procedure and the censoring mechanism for this experiment in Appendix~\ref{app:starbucks_details}.

\begin{table}[tb]
\centering
\renewcommand{\arraystretch}{1.2}
\begin{tabular}{cccccc}
\toprule
Estimator & MSE              & Squared bias    & Variance        \\ \midrule
DM        & 30.7             & 29.2            & \textbf{1.5}    \\
IPS       & 216.7            & 213.7           & 3.1             \\
DR        & 190.1            & 188.5           & \underline{1.6} \\ 
IPCW-IPS (Ours)  & \underline{11.8} & \textbf{1.4}    & 10.5            \\
IPCW-DR (Ours)  & \textbf{10.7}    & \underline{4.4} & 6.4             \\
\bottomrule
\end{tabular}
\caption{OPE results on the Starbucks dataset. The best and second-best results are highlighted in \textbf{bold} and \underline{underlined}, respectively.}
\label{tab:starbucks-ope-results}
\end{table}

\begin{table*}[t!]
\centering
\renewcommand{\arraystretch}{1.2}
\begin{tabular}{lccc}
\toprule
Learner     & RMST~(Mean $\pm$ Std) $\downarrow$ & Cost~(Mean $\pm$ Std)  & Feasible Rate $\uparrow$ \\ \midrule
Regression  & 48.86 $\pm$ 1.21                   & 2.54 $\pm$ 0.30        & 0.80                     \\
IPS         & 39.61 $\pm$ 1.90                   & 2.44 $\pm$ 0.20        & \underline{0.95}         \\
DR          & \underline{39.11 $\pm$ 2.06}       & 2.51 $\pm$ 0.23        & 0.90                     \\ 
IPCW-IPS (Ours)    & 39.41 $\pm$ 1.73                   & 2.43 $\pm$ 0.20        & \textbf{0.99}            \\
IPCW-DR (Ours)    & \textbf{38.16 $\pm$ 1.41}          & 2.53 $\pm$ 0.20        & \underline{0.95}         \\
\bottomrule
\end{tabular}
\caption{
Performance comparison of constrained OPL methods averaged over 100 trials with budget $B=2.82$.
Feasible rate denotes the proportion of trials satisfying the constraint.
The symbols $\downarrow$ and $\uparrow$ indicate that lower and higher values are better, respectively.
The best and second-best results are highlighted in \textbf{bold} and \underline{underlined}, respectively.
}
\label{tab:starbucks-constrained-opl-results}
\end{table*}

\paragraph{OPE Results.}
Table~\ref{tab:starbucks-ope-results} summarizes the OPE results on the Starbucks dataset.
As shown in Table~\ref{tab:starbucks-ope-results}, IPCW-IPS and IPCW-DR deliver a clear and substantial reduction in MSE compared with IPS and DR.  
Both IPS and DR accumulate substantial squared bias, a pattern that directly reflects the theoretical insight shown in Proposition~\ref{prop:bias_naive_estimators}.  
In contrast, IPCW-IPS and IPCW-DR control squared bias effectively, demonstrating that the proposed correction mechanism performs reliably even in realistic settings that include complex and feature-dependent censoring structures.

We can also see from the table that DM records the smallest variance, but its large bias exposes its heavy reliance on the predictive accuracy of the outcome model.  
When the outcome model is misspecified, DM suffers from substantial bias regardless of the presence of censoring.
IPCW-DR, however, counterbalances this weakness, i.e., although it shows slightly larger bias than IPCW-IPS, it significantly reduces variance and achieves the lowest overall MSE among all estimators.  
This result highlights how the doubly robust structure stabilizes the offline estimation process.  
Thanks to its doubly robust structure, IPCW-DR significantly reduces variance compared to IPCW-IPS.
These findings highlight the practical value of our estimators for real-world marketing and retention. 
By consistently outperforming baselines even under complex censoring patterns and non-linear dynamics, IPCW-IPS and IPCW-DR empower practitioners to derive effective interventional strategies from logged data that shorten purchase cycles and strengthen long-term engagement.

\paragraph{Constrained OPL Analysis.}
We further investigate the practical utility of our framework by evaluating the constrained OPL methods on the semi-synthetic Starbucks dataset.  This experiment follows the same “Oracle’’ construction used in the OPE analysis.

To set up the constrained optimization problem, we first define the cost function $c(x, a)$ based on the reward amount attached to each offer.  
In line with common performance-based marketing practices, we charge this cost only when a user converts; otherwise, the cost remains zero.  
Next, we set the budget constraint $B$ to 80\% of the average cost incurred by the logging policy ($B = 2.82$), creating a realistic and challenging setting in which the learner must deliver faster customer engagement (lower RMST) while operating under a substantially reduced spending allowance.

We compare our Constrained IPCW-IPS and Constrained IPCW-DR with three constrained baselines: Constrained Regression, Constrained IPS, and Constrained DR. 
We solve the constrained optimization problem using Lagrangian duality to enable gradient-based learning.  Appendix~\ref{app:starbucks_details} provides full details regarding the baseline implementations, neural network architecture, and optimization hyperparameters.

Table~\ref{tab:starbucks-constrained-opl-results} reports the constrained OPL results on the Starbucks dataset. 
We observe that our IPCW-DR learner delivers the strongest RMST performance and clearly outperforms both naive OPE-based learners and the model-based regression approach. 
The naive IPS and DR learners underestimate survival outcomes as they ignore the censoring issue during policy learning, and the regression baseline performs even worse because it relies entirely on an imperfect outcome model. 
In contrast, IPCW-DR combines IPCW correction with a doubly robust structure, allowing the learner to optimize survival outcomes more effectively even under strict resource constraints.
The IPCW-IPS and IPCW-DR learners also reduce the standard deviation of RMST substantially compared with naive approaches, indicating that our censoring-aware corrections create more stable gradient signals throughout the policy optimization process. 
Naive estimators unintentionally inject noise into the learning signal because they interpret censored samples as early events, leading to their unstable learning performances. 
Our IPCW-based learners avoid this distortion and update the policy using signals that reflect the underlying counterfactual survival outcomes more accurately. As a result, they repeatedly derive high-quality policies across different data samples and demonstrate strong reliability in practice.

We also examine the feasible rate, the proportion of experiment runs in which the learned policy respects the budget constraint. 
IPCW-IPS achieves the highest feasible rate of 0.99, and IPCW-DR follows closely at 0.95. 
Because the IPCW correction provides more accurate survival estimates, the learner can track the budget boundary more precisely and maintain consistent constraint satisfaction within the Lagrangian optimization framework.
Overall, these results show that our censoring-aware OPL framework not only improves survival-based outcomes but also stabilizes the learning dynamics even with challenging constraints. 
This combination of predictive accuracy, robustness, and reliability positions our approach as a practical and powerful tool for real-world decision-making regarding survival outcomes.

\section{Concluding Remarks}
This paper proposes a novel framework for Off-Policy Evaluation~(OPE) and Off-Policy Learning~(OPL) with survival outcomes subject to right censoring.
Specifically, we introduce two novel estimators, IPCW-IPS and IPCW-DR, which incorporate the statistical technique of Inverse Probability of Censoring Weighting to correct for systematic bias induced by censoring.
Theoretically, we establish the unbiasedness of our proposed estimators and prove the doubly robust property and variance reduction effects inherent in IPCW-DR.
Extensive simulation studies demonstrate that, while conventional OPE methods suffer from severe underestimation bias, our estimators maintain high estimation accuracy across a wide range of experiment conditions.
Furthermore, real-world experiments using the Starbucks dataset validate that our proposed estimators are robust against complex real-world data structures and censoring mechanisms, and substantially outperform existing baselines in terms of both OPE and OPL.

\bibliographystyle{named}
\clearpage
\appendix
\numberwithin{equation}{section}

\section{Related Work}
\label{app:related_work}
Our work builds upon three primary areas of research: survival analysis, off-policy evaluation, and constrained policy learning.

\subsection{Survival Analysis}

Survival analysis is a statistical framework designed to model survival outcomes with censoring~\cite{jenkins2005survival}. 
Foundation approaches include the Kaplan-Meier estimator~\cite{kaplan1958nonparametric} for estimating survival functions and the Cox proportional hazards model~\cite{cox1972regression} for analyzing the relationship between covariates and survival times.
Recently, deep learning-based approaches such as DeepSurv~\cite{katzman2018deepsurv} and DeepHit~\cite{lee2018deephit} have demonstrated superior predictive performance by capturing non-linear dependencies.
However, these methods primarily focus on the predictive accuracy of the survival distribution and do not directly address the off-policy estimation of policy value and optimization of decision-making policies using logged data.

Recent literature has explored the estimation of optimal treatment regimes aiming to maximize survival outcomes.
\citeauthor{jiang2017estimation}~\shortcite{jiang2017estimation} proposed value-based search methods to identify the optimal treatment regimes that maximize $t$-year survival probabilities, employing IPSW and DR estimators.
However, these estimators face critical limitations regarding censoring assumptions and optimization scalability.
First, since they typically assume independent censoring, where the censoring time is independent of covariates conditional on treatment, they can lead to bias under covariate-dependent censoring prevalent in real-world applications like customer churn.
Furthermore, their adaptation of derivative-free search algorithms~(e.g. genetic algorithms) for non-differentiable objectives limits their compatibility with modern gradient-based deep learning.
In contrast, a key advantage of our framework is that it explicitly models covariate-dependent censoring via IPCW and formulates differentiable objectives suitable for gradient-based optimization and constrained policy learning.

\subsection{Off-policy Evaluation for Incomplete Data}
Recent research has focused on adapting OPE for incomplete data.
For example, \citeauthor{wang2025offpolicy}~\shortcite{wang2025offpolicy} provide a comprehensive framework for OPE under the situation, where rewards or next-state information may be unobserved.
They demonstrate that naive estimators remain consistent under a Missing At Random~(MAR) assumption where the missing mechanism depends solely on observed information, but become biased when missingness depends on unobserved information.
Their approach is an inverse probability weighting estimator based on the dropout propensity, which they identify using a shadow variable in the challenging case of Missing Not At Random~(MNAR).

Unlike the missing data problem addressed by \citeauthor{wang2025offpolicy}~\shortcite{wang2025offpolicy}, our work focuses on the challenge of right censoring.
Right censoring is different from MNAR in two essential respects.
First, right-censoring provides partial information regarding the unobserved event time because the event is known not to have occurred before the censoring time, whereas MNAR implies a complete lack of information about the missing data.
Second, in our setting, naive estimators suffer from systematic underestimation bias due to censoring.
Crucially, this bias remains even under the random censoring assumption~(analogous to MAR), where standard OPE estimators for missing data would typically remain consistent.
This fundamental difference necessitates a novel approach specifically designed to handle the survival structure, distinct from general missing imputation.

Another related but distinct challenge in the survival context is truncation by death, addressed by \citeauthor{chu2023multiply}~\shortcite{chu2023multiply}.
They deal with the setting where an intermediate event, such as death, makes the primary outcome~(e.g., quality of life) undefined, rather than merely unobserved.
In this setting, they focus on estimating a ``survivor value function'', which is the value function for the subpopulation that would survive under any treatment.
This contrasts with our objective because their estimand focuses on a latent subpopulation, whereas our aim is to evaluate policy performance for the entire population.
Our setting therefore addresses a different and widely encountered challenge, namely right censoring in survival outcomes, which appears in applications such as maximizing subscription duration or evaluating policies for overall patient survival.

\subsection{Constrained Policy Learning}
In many operational environments, decision makers aim to maximize an outcome of interest, such as survival time and revenue, subject to a limited budget. 
This problem is fundamentally one of the resource allocation problems and has been studied in marketing and related fields~\cite{cai2023marketing,zhang2021bcorle,zhou2023direct}.
Such problems can be formulated as constrained optimization problems, where the objective is to maximize the expected reward while satisfying budget constraints.
Several approaches have used Lagrangian duality~\cite{geoffrion1974lagrangean} to find policies that maximize an objective function while satisfying constraints~\cite{ai2022lbcf,du2019improve,hao2020dynamic,zhao2019unified}.

Despite this active research on constrained OPL, existing approaches assume fully observed rewards.
To the best of our knowledge, current methods do not consider the complexities that arise when the objective is a survival outcome estimated under censoring, which leaves open an important methodological gap.
This paper fills this gap by extending our censoring aware estimators to the constrained OPL setting and provides a principled method for learning policies that aim to maximize survival duration while respecting practical budget limitations.

\section{Extensions and Theoretical Definitions}

\subsection{Estimation of Restricted Mean Survival Time~(RMST)}
\label{app:extension_RMST}

We now extend our proposed estimators to estimate the RMST, $V^{\tau}(\pi_e)$, under censoring.
Recall that the RMST is the integral of the expected survival probability $V(\pi_e, t)$ over the time horizon $[0, \tau]$. 
Based on this definition, we can construct the IPCW-IPS and IPCW-DR estimators for RMST by integrating the respective expected survival probability estimators over the same time horizon as follows.
\begin{align}
\label{eq:ipcw-ips-rmst}
\hat{V}_{\text{IPCW-IPS}}^{\tau} (\pi_e; \mathcal{D}) &\coloneq \int_{0}^{\tau} \hat{V}_{\text{IPCW-IPS}} (\pi_e, t; \mathcal{D}) dt, \\
\label{eq:ipcw-dr-rmst}
\hat{V}_{\text{IPCW-DR}}^{\tau} (\pi_e; \mathcal{D}) &\coloneq \int_{0}^{\tau} \hat{V}_{\text{IPCW-DR}} (\pi_e, t; \mathcal{D}) dt.
\end{align}

Crucially, the theoretical guarantees of the point-wise estimators naturally extend to the RMST estimation due to the linearity of the integration operator.
Specifically, the unbiasedness of the proposed estimators is preserved under the same conditions as Theorem~\ref{thm:unbiasedness}.
Moreover, the variance reduction property of IPCW-DR is maintained, since IPCW-DR achieves lower variance at each time point $t$, integrating it over the horizon $[0, \tau]$ yields an RMST estimator with consistently lower variance compared to IPCW-IPS.

\section{Detailed Experimental Settings and Additional Results}

\subsection{Data Generating Process}
\label{app:simulation_data_generating_process}

We construct the logged dataset $\mathcal{D} = \{ (x_i, a_i, T_i, r_i) \}_{i = 1}^{n}$ by generating each component sequentially as follows.
For each user $i$, the context vector $x_i \in \mathbb{R}^{10}$ is independently sampled from a standard normal distribution $N(0, I_{10})$.
The action $a_i$ is selected from a discrete action space $\mathcal{A}$ with size $|\mathcal{A}| = 10$.
The action $a_i$ is generated by the logging policy $\pi_0(a \mid x)$, which is defined as a softmax function over a linear combination of features as follows.
\begin{equation*}
\pi_0 (a \mid x) = \frac{\exp(\beta \cdot x^\top \theta_{\pi, a})}{\sum_{a^{\prime} \in \mathcal{A}} \exp(\beta \cdot x^\top \theta_{\pi, a^{\prime}})},
\end{equation*}
where $\theta_\pi \in \mathbb{R}^{d \times |\mathcal{A|}}$ is a weight matrix sampled from the uniform distribution $U(-1, 1)$, and $\beta$ is a parameter that controls the optimality and entropy of the logging policy~(set to $\beta = 1.0$ in this study).

To determine the observed time $T_i$ and the event indicator $r_i$, we explicitly simulate the true~(latent) survival time $L$ and the censoring time $C$.
The survival time $L$ follows a log-normal distribution $\log N(\mu_L(x, a), \sigma_L^2)$ with the scale parameter set to $\sigma_L = 1.0$.
The location parameter $\mu_L(x, a)$ is constructed using $\Tilde{\mu}_L(x, a)$. 
To simulate complex treatment effect heterogeneity where the impact of feature interactions varies dynamically across actions, we define $\Tilde{\mu}_L(x, a)$ as the sum of standardized linear and non-linear terms across all data as follows.
\begin{align*}
\begin{split}
\tilde{\mu}_L(x, a) &= 0.1 \cdot \theta_L^\top \phi(x, a) \\
&\quad + 5.0 \cdot (-1)^{a \pmod 2} \cdot (x_j x_k + x_m^2), \\
\mu_L(x, a) &= 1.0 \cdot \text{Standardize}(\tilde{\mu}_L(x, a)) + 0.5,
\end{split}
\end{align*}
where indices $j, k, m$ are cyclically determined based on the action $a$ as $j = (a \pmod d)$, $k = ((a + 1) \pmod d)$, and $m = ((a + 2) \pmod d)$.
The feature vector $\phi(x, a)$ using the one-hot vector $e_a$ corresponding to the action and the Kronecker product $\otimes$ is defined as follows.
\begin{equation*}
\phi(x, a) = [x^\top, e_a^\top, (x \otimes e_a)^\top]^\top \in \mathbb{R}^{d + |\mathcal{A}| + d \times |\mathcal{A}|}.
\end{equation*}

The censoring time $C$ follows an exponential distribution.
Its scale parameter $\lambda_0(x, a)$ is designed to correlate with the survival time location parameter $\mu_L(x, a)$ to introduce dependent censoring as follows.
\begin{equation*}
\log(\lambda_0(x, a)) = \theta_C^\top \phi(x, a) + \rho_0 \cdot \mu_L(x, a) + \delta_C,
\end{equation*}
where $\rho_0$ controls the correlation between survival and censoring times, and the constant $\delta_C$ is determined numerically via binary search to match the target censoring rate $\rho_1 = P(L > C)$.
We set $\rho_0$ to $-0.4$.
Finally, based on these latent variables, the observed data are determined as $T_i = \min\{ L_i, C_i \}$ and $r_i = \mathbb{I} \{ L_i \le C_i \}$.

For policy value, we adopt the RMST at time horizon $\tau = 2.0$ as follows.
\begin{equation*}
V^{\tau}(x, a) = \int_{0}^{\tau} S(t \mid x, a) dt.
\end{equation*}
The ground truth policy value $V^{\tau}(\pi)$ is approximated using a Monte Carlo simulation with $N_{test} = 100,000$ independent test samples as follows.
\begin{equation*}
V^\tau(\pi) \approx \frac{1}{N_{test}} \sum_{i=1}^{N_{test}} \sum_{a \in \mathcal{A}} \pi(a | x_i) V^\tau(x_i, a).
\end{equation*}

\subsection{Implementation Details of OPE Estimators}
\label{app:simulation_ope_details}

The proposed IPCW-IPS and IPCW-DR, as well as the baseline estimators, require the estimation of nuisance components: the propensity score $\hat{\pi}_0(a \mid x)$, the censoring model $\hat{G}(t \mid x, a)$ and the outcome model $\hat{S}(t \mid x, a)$.
In our experiments, the propensity score $\hat{\pi}_0(a \mid x)$ is estimated using multinomial logistic regression~\cite{hosmer2013applied}.
The censoring model $\hat{G}(t \mid x, a)$ and the outcome model $\hat{S}(t \mid x, a)$ are estimated using Cox proportional hazards models~\cite{cox1972regression} stratified by action.
To ensure estimation stability, particularly to prevent overfitting in samples with limited events, a small $L_2$ regularization~(with a coefficient of $10^{-4}$) is applied to the Cox proportional hazards models.

\subsection{Experimental Settings for OPL}
\label{app:simulation_opl_details}

In the OPL experiments, we parameterize the policy $\pi_\theta(a \mid x)$ using a Multi-Layer Perceptron (MLP)~\cite{lecun2015deep}.
The network architecture consists of a 10-dimensional input layer for the context vector, followed by two hidden layers with 64 units each with ReLU activation functions~\cite{nair2010rectified}.
The output layer corresponds to the action space size (10 dimensions) and applies a softmax activation to ensure a valid probability distribution.
Network weights are initialized using Xavier Uniform initialization~\cite{xavier2010understanding}, while bias terms are initialized to zero.

We implement two distinct learning paradigms.
\paragraph{Regression-based Approach}
This approach acts as a model-based baseline using the estimated outcome model $\hat{S}(t \mid x, a)$.
The learner calculates the predicted RMST for every action through numerical integration using the trapezoidal rule over $M = 100$ equally spaced grid points within the interval $[0, \tau]$ as follows.
\begin{equation*}
\hat{V}^\tau(x, a) \approx \sum_{j=1}^{M} \hat{S}(t_j|x, a) \Delta t,
\end{equation*}
where $\Delta t = \tau / M$. The learner then selects the action that maximizes this predicted value in a deterministic manner.

\paragraph{Policy Gradient-based Approaches}
These learners treat the OPE estimators as differentiable objective functions. They optimize the policy parameters $\theta$ to maximize the following objective:
\begin{equation*}
J(\theta) = \hat{V}_{OPE}(\pi_\theta; \mathcal{D}) \approx \frac{1}{n} \sum_{i=1}^{n} \hat{v}_i(\pi_\theta),
\end{equation*}
where $\hat{v}_i(\pi_\theta)$ represents the individual sample contribution to the estimator (e.g., the IPCW-weighted reward).
We use the Adam optimizer~\cite{kingma2017adam} with a learning rate of $0.01$ and a batch size of 256. The training runs for up to 200 epochs.
To prevent overfitting, we employ an Early Stopping protocol~\cite{goodfellow2016deep} based on a validation set comprising 30\% of the training data, with a patience of 10 epochs.

\paragraph{Evaluation Protocol}
The performance of each learned policy $\hat{\pi}$ is quantitatively evaluated on an independent large-scale test set ($\mathcal{D}_{test}$) containing $N_{test} = 100,000$ contexts $x_j$ sampled from the standard normal distribution.
The final performance metric is the improvement ratio relative to the logging policy, defined as $Ratio = V^\tau(\hat{\pi}) / V^\tau(\pi_0)$.
The true policy value is approximated via Monte Carlo simulation using the true underlying parameters as follows:
\begin{equation*}
V^\tau(\hat{\pi}) \approx \frac{1}{N_{test}} \sum_{j=1}^{N_{test}} \sum_{a \in \mathcal{A}} \hat{\pi}(a \mid x_j) V^\tau (x_j, a).
\end{equation*}

\subsection{Impact of Logging Policy's Quality}
\label{app:impact_beta}

We investigate the robustness of the OPL methods against the quality of the logging policy $\pi_0$.
We control the behavior of the logging policy by varying the parameter $\beta$ defined in Appendix~\ref{app:simulation_data_generating_process}.
A smaller $\beta$ yields a more uniform~(random) logging policy with higher entropy, whereas a larger $\beta$ results in a more deterministic policy that prefers specific actions.
We vary $\beta \in \{ 0.0, 0.5, 1.0, 1.5, 2.0 \}$ while keeping the other experimental conditions fixed.

Figure~\ref{fig:opl-beta} illustrates the improvement ratio of the learned policies compared to the logging policy.
\begin{figure}[!t]
\centering
\includegraphics[width=\linewidth]{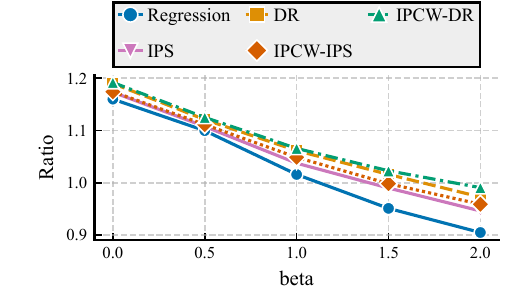} \vspace{-4mm}
\caption{
OPL performance comparison under varying $\beta$.
The improvement ratio $V^\tau(\hat{\pi}) / V^\tau(\pi_0)$ quantifies how effectively each method learns a superior policy compared to the logging policy optimalities.
}
\label{fig:opl-beta}
\end{figure}
The results demonstrate that our proposed estimators, particularly IPCW-DR, outperform the naive baselines and the regression-based approach across all values of $\beta$.

\subsection{Data Construction for Semi-Synthetic Experiments}
\label{app:starbucks_details}

In the real-world application, we use the Starbucks Customer Reward Program dataset to construct a semi-synthetic environment.
This section details the data preprocessing, the construction of the Oracle environment including non-linear reward modifications, and the experimental setup for comparing OPE estimators.

\paragraph{Oracle Environment Construction}
We construct an Oracle environment from the Starbucks data so that the true RMST of every action is available for every user.
We first remove records with invalid values, such as an age of 118.
Each offer reception then forms an episode, whose survival time is the time until the first transaction within the validity period of the offer.
An episode without a transaction is censored at the end of the validity period.
We split the processed data into an environment construction set $\mathcal{D}_{env}$ with 60\% of the episodes and a pooling set $\mathcal{D}_{pool}$ with the remaining 40\%.
On $\mathcal{D}_{env}$, we train a Random Survival Forest (RSF)~\cite{ishwaran2008random} with 100 trees, a maximum depth of 7, and a minimum leaf size of 10.
The RSF takes as input the standardized user features, the offer features, and the action.
The user features are age, income, membership tenure, and gender, and the offer features are difficulty, reward, duration, and delivery channels.
The oracle RMST $V^\tau(x, a)$ is the integral of the RSF survival function up to $\tau = 120$ hours.

To ensure that the environment has dynamics that simple linear models cannot capture, we add a non-linear interaction to the oracle RMST.
For action $a = 0$, we multiply the RMST by 0.2 for male users whose income exceeds \$80,000 and whose age is below 50.
The RMST of all other combinations of users and actions remains the RSF prediction.

\paragraph{Artificial Censoring Mechanism}
We generate artificial censoring times that depend on user demographics, so that censoring depends on the covariates as described in the main text.
The censoring time $C_i$ follows an exponential distribution with rate
\begin{equation*}
\lambda_i = \lambda_0 \cdot \exp(-\beta_1 \bar{x}_{\text{income}} - \beta_2 \bar{x}_{\text{age}}),
\end{equation*}
where $\bar{x}_{\text{income}}$ and $\bar{x}_{\text{age}}$ are the income and age standardized within each sample.
We set $\lambda_0 = 1/120$, which gives a mean censoring time of 120 hours to a user with average income and age.
We set the coefficients to $\beta_1 = 0.8$ and $\beta_2 = 0.3$.
The observed time and the event indicator are then updated as $T_i^* = \min\{ T_i, C_i \}$ and $r_i^* = \mathbb{I}\{ T_i \leq C_i \}$.

\paragraph{Experimental Setup and Model Misspecification}
Each trial evaluates an $\epsilon$-greedy target policy on $N = 5,000$ episodes sampled from $\mathcal{D}_{pool}$.
The logging policy $\pi_0$ is the action selection distribution observed in the dataset, which is approximately uniform.
The target policy $\pi_e$ is an $\epsilon$-greedy policy with $\epsilon = 0.1$ around the action that minimizes the oracle RMST.
The ground-truth policy value is the average of $\sum_{a} \pi_e(a \mid x_i) V^\tau(x_i, a)$ over the sampled users.

We intentionally restrict the nuisance models to evaluate the estimators under model misspecification.
The Oracle environment relies on a non-linear RSF that uses all available attributes.
In contrast, the propensity score model is a multinomial logistic regression, and the censoring and outcome models are Cox proportional hazards models fitted separately for each action.
These nuisance models use only user demographics and exclude offer attributes, which induces confounding bias.

\paragraph{Experimental Settings for Constrained Off-Policy Learning}
The constrained OPL experiment shares the RSF oracle and the artificial censoring mechanism with the OPE analysis and adapts only the non-linear modification to the cost-constrained setting.
For the offer with the largest reward, which is action $a = 0$, we multiply the RMST by 0.1 for users whose income exceeds \$65,000 or whose age is below 55.
We apply the same modification to the observed survival times of the training episodes.
The cost oracle $c(x, a)$ is a random forest regressor trained on $\mathcal{D}_{env}$ to predict the incurred reward.
This regressor has 100 trees, a maximum depth of 10, and a minimum leaf size of 5.

Each of the 100 trials trains the learners on newly sampled episodes and evaluates them on a fixed set of users.
In each trial, we sample 5,000 training episodes from $\mathcal{D}_{pool}$.
The evaluation set consists of 5,000 users drawn once from $\mathcal{D}_{pool}$, and the learned policies are evaluated on it with the oracle RMST and cost.
We set the budget $B$ to 80\% of the average oracle cost of the logged actions on the evaluation set, which gives $B = 2.82$.

We compare the proposed methods with three baselines, namely Constrained Regression, Constrained IPS, and Constrained DR.
Constrained Regression is a model-based approach that predicts the RMST and the expected cost of each action.
It then selects the action with the smallest predicted RMST among those satisfying the budget constraint.
Constrained IPS and Constrained DR are gradient-based primal-dual methods that optimize the objective estimated by the IPS and DR estimators, respectively.

All gradient-based approaches, including the baselines and our proposed methods, share the same policy model and optimizer.
The policy $\pi_\theta(a \mid x)$ is an MLP with two hidden layers of 64 units and ReLU activations.
We optimize the network weights with the Adam optimizer, using a learning rate of $1.0 \times 10^{-3}$, a weight decay of $10^{-4}$, and a batch size of 512.
We also clip the gradient norm at 1.0 to stabilize the updates.

The Lagrange multiplier and the model selection both use a validation split so that the learned policy satisfies the budget.
We hold out 30\% of the training episodes as the validation split.
After every epoch, we update the Lagrange multiplier $\lambda$ with a learning rate of 0.05 based on the cost estimated on the validation split.
We initialize $\lambda$ to 3.0, which encourages the learner to stay within the feasible region from the start.
To satisfy the budget robustly despite stochastic gradient updates, the update of $\lambda$ targets a slightly stricter budget $B - 0.05$.
We train the models for 500 epochs and keep the parameters of the epoch with the smallest estimated RMST among the epochs whose estimated validation cost satisfies the budget.
Finally, we evaluate the resulting deterministic greedy policy on the evaluation set and average the results over the 100 trials.

\section{Proofs of Theoretical Results} \label{app:proof}

\subsection{Proof of Proposition~\ref{prop:bias_naive_estimators}}
\label{app:bias_naive_estimators}
\begin{proof}
Let $w(x, a) = \pi_e(a \mid x) / \pi_0(a \mid x)$.
The expectation of the naive IPS estimator is calculated as follows.
\begin{align*}
&\mathbb{E}_{p(\mathcal{D})} \left[ \hat{V}_{\text{IPS} } (\pi_e, t; \mathcal{D}) \right] \\
&=\mathbb{E}_{p(x) \pi_0(a \mid x) p(L, C \mid x, a)} [w(x, a) \mathbb{I} \{ T > t\} ] \\
&= \mathbb{E}_{p(x) \pi_0(a \mid x)} \left[ w(x, a) \mathbb{E}_{p(L, C \mid x, a)} [ \mathbb{I} \{ \min\{L, C \} > t \} ] \right] \\
&= \mathbb{E}_{p(x) \pi_0(a \mid x)} \biggl[ w(x, a) \mathbb{E}_{p(L \mid x, a)} [\mathbb{I} \{ L > t \}] \\
&\quad\quad\quad\quad\quad \times \mathbb{E}_{p(C \mid x, a)} [\mathbb{I} \{ C > t \}] \biggr] \\
&= \mathbb{E}_{p(x) \pi_0(a \mid x)} \left[ w(x, a) P(L > t \mid x, a) P(C > t \mid x, a) \right] \\
&= \mathbb{E}_{p(x) \pi_0(a \mid x)} \left[ w(x, a) S(x, a, t) G(t \mid x, a) \right] \\
&= \mathbb{E}_{p(x)} \left[ \sum_{a \in \mathcal{A}} \pi_0(a \mid x) w(x, a) S(x, a, t) G(t \mid x, a) \right] \\
&= \mathbb{E}_{p(x)} \left[ \sum_{a \in \mathcal{A}} \pi_0(a \mid x) \frac{\pi_e(a \mid x)}{\pi_0(a \mid x)} S(x, a, t) G(t \mid x, a) \right] \\
&= \mathbb{E}_{p(x)} \left[ \sum_{a \in \mathcal{A}} \pi_e(a \mid x) S(x, a, t) G(t \mid x, a) \right] \\
&= \mathbb{E}_{p(x) \pi_e(a \mid x)} [ S(x, a, t) G(t \mid x, a) ].
\end{align*}

First, we define the naive doubly robust estimator as
\begin{align}
    & \hat{V}_{\mathrm{DR}}(\pi_e,t;\mathcal{D}) = \frac{1}{n}\sum_{i=1}^n \Bigg\{w(x_i, a_i) (\mathbb{I}\{T_i > t\}-\hat{S}(x_i, a_i, t)) \notag \\
    & \hspace{5cm} + \sum_{a \in \mathcal{A}} \pi_e(a \mid x_i)\,\hat{S}(x_i, a, t) \Bigg\} \label{eq:DR} 
\end{align}

The expectation of the naive DR estimator $\hat{V}_{\text{DR}}$ is calculated as follows.
\begin{align*}
&\mathbb{E}_{p(\mathcal{D})} \left[ \hat{V}_{\text{DR}} (\pi_e, t; \mathcal{D}_H) \right] \\
&=\mathbb{E}_{p(x) \pi_0(a \mid x)} [ w(x, a) \mathbb{E}_{p(L, C \mid x, a)} [ \mathbb{I} \{T_i > t\} ] ] \\ 
&\quad- \mathbb{E}_{p(x) \pi_0(a \mid x) p(L, C \mid x, a)} \left[ w(x, a) \hat{S} (x_i, a_i, t) \right] \\
&\quad+ \mathbb{E}_{p(x) \pi_0(a \mid x) p(L, C \mid x, a)} \left[ \sum_{a \in \mathcal{A}} \pi_e(a \mid x_i) \hat{S}(x_i, a_i, t) \right] \\
&= \mathbb{E}_{p(x) \pi_e(a \mid x)} [ S(x, a, t) G(t \mid x, a) ] \\
&\quad- \mathbb{E}_{p(x) \pi_e(a \mid x)} [ \hat{S} (x_i, a_i, t) ] \\
&\quad+ \mathbb{E}_{p(x)} \left[ \sum_{a \in \mathcal{A}} \pi_e(a \mid x_i) \hat{S}(x_i, a_i, t) \right] \\
&= \mathbb{E}_{p(x) \pi_e(a \mid x)} [ S(x, a, t) G(t \mid x, a) ] \\
&\quad- \mathbb{E}_{p(x) \pi_e(a \mid x)} [ \hat{S} (x_i, a_i, t) ] \\
&\quad+ \mathbb{E}_{p(x) \pi_e(a \mid x)} [ \hat{S}(x_i, a_i, t) ] \\
&= \mathbb{E}_{p(x) \pi_e(a \mid x)} [ S(x, a, t) G(t \mid x, a) ].
\end{align*}
Therefore, the bias of the naive estimators can be summarized as follows.
\begin{equation*}
  \text{Bias} [\hat{V}_{\text{IPS}}] = \text{Bias} [\hat{V}_{\text{DR}}] = \mathbb{E}_{p(x) \pi_e(a \mid x)} [ S(x, a, t) ( G(t \mid x, a) - 1 ) ].
\end{equation*}

\end{proof}

\subsection{Proof of Theorem~\ref{thm:unbiasedness}}
\label{app:unbiasedness}
\begin{proof}
We calculate the expectation of IPCW-IPS estimator.
\begin{align*}
&\mathbb{E}_{p(\mathcal{D})} [\hat{V}_{\text{IPCW-IPS}} (\pi_e, t)] \\
&= \mathbb{E}_{p(x) \pi_0(a \mid x) } \left[ \frac{w(x, a)}{G(t \mid x, a)} \mathbb{E}_{p(L, C \mid x, a)} [\mathbb{I} \{T > t \} ] \right] \\
&= \mathbb{E}_{p(x) \pi_0(a \mid x) } \left[ \frac{w(x, a)}{G(t \mid x, a)} S(x, a, t) G(t \mid x, a) \right] \\ 
&\quad (\because \text{Proof of Proposition~\ref{prop:bias_naive_estimators}}) \\
&= \mathbb{E}_{p(x) \pi_0(a \mid x) } [ w(x, a) S(x, a, t)  ] \\ 
&= \mathbb{E}_{p(x) \pi_e(a \mid x) } [ S(x, a, t) ] \\ 
&= V(\pi_e, t).
\end{align*}
Therefore, the IPCW-IPS estimator is unbiased for the expected survival probability $V(\pi_e, t)$ if the propensity score $\pi_0(a \mid x)$ is correctly specified.

Next, we calculate the expectation of IPCW-DR estimator as follows.
\begin{align*}
&\mathbb{E}_{p(\mathcal{D})} [\hat{V}_{\text{IPCW-DR}} (\pi_e, t)] \\
&=\mathbb{E}_{p(x) \pi_0(a \mid x)} \Biggl[ w(x, a) \biggl( \mathbb{E}_{p(L, C \mid x, a)} \left[ \frac{\mathbb{I} \{ T > t \} }{G(t \mid x, a)} \right] \\
&\quad\quad - \hat{S}(x, a, t) \biggr) + \sum_{a \in \mathcal{A}} \pi_e(a \mid x) \hat{S}(x, a, t) \Biggr] \\
&\quad+ \sum_{a \in \mathcal{A}} \pi_e(a \mid x) \hat{S}(x, a, t) \Biggr] \\
&=\mathbb{E}_{p(x) \pi_0(a \mid x)} \Biggl[ w(x, a) ( S(x, a, t) - \hat{S}(x, a, t) ) \Biggr] \\
&\quad+ \mathbb{E}_{p(x)} \left[ \sum_{a \in \mathcal{A}} \pi_e(a \mid x) \hat{S}(x, a, t) \right].
\end{align*}
\paragraph{(i) Case when the propensity score $\pi_0(a \mid x)$ is correctly specified:}
\begin{align*}
&\mathbb{E}_{p(x) \pi_0(a \mid x)} \Biggl[ w(x, a) ( S(x, a, t) - \hat{S}(x, a, t) ) \Biggr] \\
&\quad+ \mathbb{E}_{p(x)} \left[ \sum_{a \in \mathcal{A}} \pi_e(a \mid x) \hat{S}(x, a, t) \right] \\
&=\mathbb{E}_{p(x)} \left[ \sum_{a \in \mathcal{A}} \pi_0(a \mid x) \frac{\pi_e(a \mid x)}{\pi_0(a \mid x)} (S(x, a, t) - \hat{S}(x, a, t)) \right] \\
&\quad+ \mathbb{E}_{p(x)} \left[ \sum_{a \in \mathcal{A}} \pi_e(a \mid x) \hat{S}(x, a, t) \right] \\
&=\mathbb{E}_{p(x) \pi_e(a \mid x)} [S(x, a, t)] - \mathbb{E}_{p(x) \pi_e(a \mid x)} [\hat{S}(x, a, t)] \\
&\quad+ \mathbb{E}_{p(x) \pi_e(a \mid x)} [\hat{S}(x, a, t)] \\
&= V(\pi_e, t).
\end{align*}

\paragraph{(ii) Case when the outcome model $\hat{S}(x, a, t)$ is correctly specified:}
\begin{align*}
&\mathbb{E}_{p(x) \pi_0(a \mid x)} \Biggl[ w(x, a) ( S(x, a, t) - \hat{S}(x, a, t) ) \Biggr] \\
&\quad+ \mathbb{E}_{p(x)} \left[ \sum_{a \in \mathcal{A}} \pi_e(a \mid x) \hat{S}(x, a, t) \right] \\
&=\mathbb{E}_{p(x) \pi_0(a \mid x)} \Biggl[ w(x, a) ( S(x, a, t) - S(x, a, t) ) \Biggr] \\
&\quad+ \mathbb{E}_{p(x)} \left[ \sum_{a \in \mathcal{A}} \pi_e(a \mid x) S(x, a, t) \right] \\
&= \mathbb{E}_{p(x) \pi_e(a \mid x)} [S(x, a, t)] \\
&= V(\pi_e, t).
\end{align*}
Therefore, in either case, the IPCW-DR estimator is unbiased for the expected survival probability $V(\pi_e, t)$.

\end{proof}

\subsection{Proof of Theorem~\ref{thm:variance}}
\label{app:variance}

We provide the detailed derivation of the variance of IPCW-IPS and IPCW-DR estimators and formally prove the variance reduction property of IPCW-DR. 

\begin{proof}
We first derive the variance of the IPCW-IPS estimator.
The variance of the IPCW-IPS estimator is given by
\begin{align*}
&n \mathbb{V}_{p(\mathcal{D})} [\hat{V}_{\text{IPCW-IPS}} (\pi_e, t; \mathcal{D})] \\
&=\mathbb{V}_{p(x) \pi_0(a \mid x) p(L, C \mid x, a)} \left[ w(x, a) \frac{\mathbb{I} \{ T > t \}}{G(t \mid x, a)} \right] \\
&= \mathbb{E}_{p(x) \pi_0(a \mid x)} \left[ \mathbb{V}_{p(L, C \mid x, a)} \left[ w(x, a) \frac{\mathbb{I} \{ T > t \} }{G(t \mid x, a)} \right] \right] \\
&\quad+ \mathbb{V}_{p(x) \pi_0(a \mid x)} \left[ \mathbb{E}_{p(L, C \mid x, a)} \left[ w(x, a) \frac{\mathbb{I} \{ T > t \}}{G(t \mid x, a)} \right] \right] \\
&= \mathbb{E}_{p(x) \pi_0(a \mid x)} \left[ (w(x, a))^2 \mathbb{V}_{p(L, C \mid x, a)} \left[  \frac{\mathbb{I} \{ T > t \} }{G(t \mid x, a)} \right] \right] \\
&\quad+ \mathbb{V}_{p(x) \pi_0(a \mid x)} [ w(x, a) S(x, a, t) ] \\
&= \mathbb{E}_{p(x) \pi_0(a \mid x)} \left[ (w(x, a))^2 \mathbb{V}_{p(L, C \mid x, a)} \left[  \frac{\mathbb{I} \{ T > t \} }{G(t \mid x, a)} \right] \right] \\
&\quad+ \mathbb{E}_{p(x)} [\mathbb{V}_{\pi_0(a \mid x)} [w(x, a) S(x, a, t)]] \\
&\quad+ \mathbb{V}_{p(x)} [\mathbb{E}_{\pi_0(a \mid x)} [w(x, a) S(x, a, t)]] \\
&= \mathbb{E}_{p(x) \pi_0(a \mid x)} \left[ (w(x, a))^2 \mathbb{V}_{p(L, C \mid x, a)} \left[  \frac{\mathbb{I} \{ T > t \} }{G(t \mid x, a)} \right] \right] \\
&\quad+ \mathbb{E}_{p(x)} [\mathbb{V}_{\pi_0(a \mid x)} [w(x, a) S(x, a, t)]] + \mathbb{V}_{p(x)} [S^{\pi_e}(x, t)] \\
&= \mathbb{E}_{p(x) \pi_0(a \mid x)} [(w(x, a))^2 \sigma^2(x, a, t)] \\
&\quad+ \mathbb{E}_{p(x)} [\mathbb{V}_{\pi_0(a \mid x)} [w(x, a) S(x, a, t)]] + \mathbb{V}_{p(x)} [S^{\pi_e}(x, t)],
\end{align*}
where 
\begin{gather*}
S^{\pi_e}(x, t) \coloneq \mathbb{E}_{\pi_e(a \mid x)} [S(x, a, t)], \\ 
\sigma^2(x, a, t) \coloneq \mathbb{V}_{p(L, C \mid x, a)} \left[ \frac{\mathbb{I} \{T > t\}} {G(t \mid x, a)} \right].
\end{gather*}

Next, we derive the variance of the IPCW-DR estimator.
Assume that the $\hat{S}(x, a, t)$ and $\hat{G}(t \mid x, a)$ are correctly specified.
To simplify the derivation, let $\delta(x, a, t, T) \coloneq \frac{\mathbb{I} \{ T > t \} }{G(t \mid x, a)} - S(x, a, t)$ denote the IPCW-adjusted residual.
Using this notation, the variance of the IPCW-DR estimator is given by
\begin{align}
\begin{split}
\label{app-eq:variance_IPCW-DR}
&n \mathbb{V}_{p(\mathcal{D})} [\hat{V}_{\text{IPCW-DR}} (\pi_e, t; \mathcal{D})] \\
&= \mathbb{V}_{p(x) \pi_0(a \mid x) p(L, C \mid x, a)} \biggl[ w(x, a) \delta(x, a, t, T) \\
&\quad\quad + \sum_{a \in \mathcal{A}} \pi_e(a \mid x) S(x, a, t) \biggr] \\
&= \mathbb{E}_{p(x) \pi_0(a \mid x)} \biggl[\mathbb{V}_{p(L, C \mid x, a)} \biggl[ w(x, a) \delta(x, a, t, T) \\
&\quad\quad + \sum_{a \in \mathcal{A}} \pi_e(a \mid x) S(x, a, t) \biggr] \biggr] \\
&\quad+ \mathbb{V}_{p(x) \pi_0(a \mid x)} \biggl[\mathbb{E}_{p(L, C \mid x, a)} \biggl[ w(x, a) \delta(x, a, t, T) \\
&\quad\quad + \sum_{a \in \mathcal{A}} \pi_e(a \mid x) S(x, a, t) \biggr] \biggr].
\end{split}
\end{align}

We first focus on the inner expectation $\mathbb{E}_{p(L, C \mid x, a)} [ \cdot ]$:
\begin{align*}
&\mathbb{E}_{p(L, C \mid x, a)} \biggl[ w(x, a) \delta(x, a, t, T) \\
&\quad\quad + \sum_{a \in \mathcal{A}} \pi_e(a \mid x) S(x, a, t) \biggr] \\
&= w(x, a) \mathbb{E}_{p(L, C \mid x, a)} [ \delta(x, a, t, T) ] \\
&\quad\quad + \sum_{a \in \mathcal{A}} \pi_e(a \mid x) S(x, a, t) \\
&= w(x, a) \underbrace{ \left( S(x, a, t) - S(x, a, t) \right) }_{=0} \\
&\quad\quad + \sum_{a \in \mathcal{A}} \pi_e(a \mid x) S(x, a, t) \\
&= \sum_{a \in \mathcal{A}} \pi_e(a \mid x) S(x, a, t).
\end{align*}

Next, we focus on the inner variance $\mathbb{V}_{p(L, C \mid x, a)} [ \cdot ]$:
\begin{align*}
&\mathbb{V}_{p(L, C \mid x, a)} \Biggl[ w(x, a) \left( \frac{\mathbb{I} \{ T > t \} }{G(t \mid x, a)} - S(x, a, t) \right) \\
&\quad+ \sum_{a \in \mathcal{A}} \pi_e(a \mid x) S(x, a, t) \Biggr] \\
&= (w(x, a))^2 \mathbb{V}_{p(L, C \mid x, a)} \left[ \frac{\mathbb{I} \{ T > t \} }{G(t \mid x, a)}  \right] \\
&= (w(x, a))^2 \sigma^2 (x, a, t).
\end{align*}

Substituting these results into Equation~\eqref{app-eq:variance_IPCW-DR}, we obtain:
\begin{align*}
&\mathbb{E}_{p(x) \pi_0(a \mid x)} [(w(x, a))^2 \sigma^2 (x, a, t)] \\
&\quad+ \mathbb{V}_{p(x) \pi_0(a \mid x)} \left[ \mathbb{E}_{\pi_e(a \mid x)} [S(x, a, t)] \right] \\
&=\mathbb{E}_{p(x) \pi_0(a \mid x)} [(w(x, a))^2 \sigma^2 (x, a, t)] \\
&\quad+ \mathbb{V}_{p(x)} [ \mathbb{E}_{\pi_0(a \mid x)} [ \mathbb{E}_{\pi_e(a \mid x)} [S(x, a, t)]] ] \\
&\quad+ \mathbb{E}_{p(x)} [\mathbb{V}_{\pi_0(a \mid x)} [\mathbb{E}_{\pi_e(a \mid x)} [S(x, a, t)]]] \\
&=\mathbb{E}_{p(x) \pi_0(a \mid x)} [(w(x, a))^2 \sigma^2 (x, a, t)] + \mathbb{V}_{p(x)} [ S^{\pi_e} (x, t) ].
\end{align*}

Finally, we compare the variances of the two estimators to prove the variance reduction property.
Subtracting the variance of IPCW-DR from that of IPCW-IPS is given by
\begin{align*}
&n \mathbb{V}_{p(\mathcal{D})} [\hat{V}_{\text{IPCW-IPS}}] - n \mathbb{V}_{p(\mathcal{D})} [\hat{V}_{\text{IPCW-DR}}] \\
&= \mathbb{E}_{p(x) \pi_0(a \mid x)} [(w(x, a))^2 \sigma^2(x, a, t)] + \mathbb{V}_{p(x)} [S^{\pi_e}(x, t)] \\
&\quad+ \mathbb{E}_{p(x)} [\mathbb{V}_{\pi_0(a \mid x)} [w(x, a) S(x, a, t)]] \\
&\quad- \Bigl( \mathbb{E}_{p(x) \pi_0(a \mid x)} [(w(x, a))^2 \sigma^2 (x, a, t)] + \mathbb{V}_{p(x)} [ S^{\pi_e} (x, t) ] \Bigr) \\
&= \mathbb{E}_{p(x)} [ \mathbb{V}_{\pi_0(a \mid x)} [w(x, a) S(x, a, t) ] ] \\
& \geq 0. 
\end{align*}
This confirms that IPCW-DR achieves a variance reduction equal to the expected conditional variance of the weighted outcome model, thereby ensuring $\mathbb{V}[\hat{V}_{\text{IPCW-DR}}] \le \mathbb{V}[\hat{V}_{\text{IPCW-IPS}}]$.

\end{proof}

\subsection{Theoretical Analysis of IPCW-IPS Variance Properties}
\label{app:variance_properties}
We provide a detailed analysis of the variance of the IPCW-IPS estimator.
First, we analyze the variance component $\sigma^2(x, a, t)$ defined as follows.
\begin{align*}
\sigma^2(x, a, t) &\coloneq \mathbb{V}_{p(L, C \mid x, a)} \left[ \frac{\mathbb{I} \{T > t\}}{G(t \mid x, a)} \right] \\
&= \mathbb{V}_{p(L, C \mid x, a)} \left[ \frac{\mathbb{I} \{ \min\{L, C\} > t \}}{G(t \mid x, a)} \right] \\
&= \mathbb{E}_{p(L, C \mid x, a)} \left[ \frac{(\mathbb{I} \{ \min\{L, C\} > t \})^2}{(G(t \mid x, a))^2} \right] \\
&\quad- \left( \mathbb{E}_{p(L, C \mid x, a)} \left[ \frac{\mathbb{I} \{ \min\{L, C\} > t \}}{G(t \mid x, a)} \right] \right)^2  \\
&= \frac{S(x, a, t)}{G(t \mid x, a)} - ( S(x, a, t) )^2.
\end{align*}

This decomposition reveals a critical property of the variance in IPCW-IPS estimator.
The variance is inversely proportional to the censoring survival function $G(t \mid x, a)$.
If censoring occurs frequently such that $G(t \mid x, a)$ is small, the variance $\sigma^2(x, a, t)$ becomes large.

\subsection{Proof of Unbiasedness of Equation~(\ref{eq:ipcw-ips-rmst}) and (\ref{eq:ipcw-dr-rmst})}
\begin{proof}
First, we prove the unbiasedness of the IPCW-IPS estimator for RMST.
We calculate its expectation as follows.
\begin{align*}
&\mathbb{E}_{p(\mathcal{D})} \left[ \hat{V}_{\text{IPCW-IPS}}^{\tau} (\pi_e; \mathcal{D}) \right] \\
&= \mathbb{E}_{p(\mathcal{D})} \left[ \int_{0}^{\tau} \hat{V}_{\text{IPCW-IPS}} (\pi_e, t; \mathcal{D}) dt \right] \\
&= \int_{0}^{\tau} \mathbb{E}_{p(\mathcal{D})} \left[ \hat{V}_{\text{IPCW-IPS}} (\pi_e, t; \mathcal{D}) \right] dt \\
&= \int_{0}^{\tau} V(\pi_e, t) dt \\
&\quad (\because \text{Theorem~\ref{thm:unbiasedness}}) \\
&= V^{\tau}(\pi_e).
\end{align*}

Next, we prove the unbiasedness of the IPCW-DR estimator for RMST.
We calculate its expectation as follows.
\begin{align*}
&\mathbb{E}_{p(\mathcal{D})} \left[ \hat{V}_{\text{IPCW-DR}}^{\tau} (\pi_e; \mathcal{D}) \right] \\
&= \mathbb{E}_{p(\mathcal{D})} \left[ \int_{0}^{\tau} \hat{V}_{\text{IPCW-DR}} (\pi_e, t; \mathcal{D}) dt \right] \\
&= \int_{0}^{\tau} \mathbb{E}_{p(\mathcal{D})} \left[ \hat{V}_{\text{IPCW-DR}} (\pi_e, t; \mathcal{D}) \right] dt \\
&= \int_{0}^{\tau} V(\pi_e, t) dt \\
&\quad (\because \text{Theorem~\ref{thm:unbiasedness}}) \\
&= V^{\tau}(\pi_e).
\end{align*}
\end{proof}

\subsection{Derivation of Equation~\eqref{eq:policy_gradient}}

\begin{align*}
&\nabla_{\theta}V(\pi_{\theta}, t) \\
&=\nabla_{\theta} \mathbb{E}_{p(x) \pi_{\theta}(a \mid x)} [S(x, a, t)] \\
&=\mathbb{E}_{p(x)} \left[ \sum_{a \in \mathcal{A}} S(x, a, t) \nabla_{\theta} \pi_{\theta} (a \mid x) \right] \\
&=\mathbb{E}_{p(x)} \left[ \sum_{a \in \mathcal{A}} S(x, a, t) \pi_{\theta}(a \mid x) \nabla_{\theta} \log \pi_{\theta} (a \mid x) \right] \\
&=\mathbb{E}_{p(x) \pi_{\theta}(a \mid x)} \left[ S(x, a, t) \nabla_{\theta} \log \pi_{\theta} (a \mid x) \right]
\end{align*}

\subsection{Proof of Unbiasedness of IPCW-IPS and IPCW-DR against the True Lagrangian Policy Gradient}
\label{app:estimators-against-Lagrangian}

\begin{proof}
We calculate the expectation of the IPCW-IPS estimator for the Lagrangian policy gradient as follows:
\begin{align*}
&\mathbb{E}_{p(\mathcal{D})} \left[ \widehat{\nabla_{\theta} L}_{\text{IPCW-IPS}} \right] \\
&= \mathbb{E}_{p(x) \pi_0(a \mid x) p(L, C \mid x, a)} \Biggl[ w^{\prime} (x, a) \left( \frac{\mathbb{I} \{ T > t \}}{G(t \mid x, a)} - \lambda c(x, a) \right) \\
&\quad\quad \times \nabla_{\theta} \log \pi_{\theta}(a \mid x) \Biggr] \\
&= \mathbb{E}_{p(x) \pi_0(a \mid x) } \Biggl[ w^{\prime} (x, a) \nabla_{\theta} \log \pi_{\theta}(a \mid x) \\
&\quad\quad \times ( S(x, a, t) - \lambda c(x, a) ) \Biggr] \\
&= \mathbb{E}_{p(x) \pi_0(a \mid x) } \Biggl[ w^{\prime} (x, a) \nabla_{\theta} \log \pi_{\theta}(a \mid x) R_{\lambda}(x, a, t) \Biggr] \\
&= \mathbb{E}_{p(x) } \Biggl[ \sum_{a \in \mathcal{A} }\pi_0(a \mid x) \frac{\pi_\theta(a \mid x)}{\pi_0(a \mid x)} \nabla_{\theta} \log \pi_{\theta}(a \mid x) R_{\lambda}(x, a, t) \Biggr] \\
&= \mathbb{E}_{p(x) \pi_\theta (a \mid x) } [ R_{\lambda}(x, a, t) \nabla_{\theta} \log \pi_{\theta}(a \mid x) ] \\
&= \nabla_{\theta} L(\pi_{\theta}, t, \lambda).
\end{align*}
Therefore, the IPCW-IPS estimator for the Lagrangian policy gradient is an unbiased estimator of the true Lagrangian policy gradient $\nabla_{\theta} L(\pi_{\theta}, t, \lambda)$.
\end{proof}

\begin{proof}
We calculate the expectation of the IPCW-DR estimator for the Lagrangian policy gradient as follows:
\begin{align*}
&\mathbb{E}_{p(\mathcal{D})} \left[ \widehat{\nabla_{\theta} L}_{\text{IPCW-DR}} \right] \\
&=\mathbb{E}_{p(x) \pi_0(a \mid x) p(L, C \mid x, a)} \Biggl[ w^{\prime}(x, a) \nabla_{\theta} \log \pi_{\theta}(a \mid x) \\
&\quad\quad \times \Biggl( \frac{\mathbb{I} \{ T > t \}}{G(t \mid x, a)} - \hat{S}(x, a, t) \Biggr) \\
&\quad\quad + \mathbb{E}_{\pi_{\theta}(a \mid x)} [ ( \hat{S}(x, a, t) - \lambda c(x, a) ) \nabla_{\theta} \log \pi_{\theta}(a \mid x) ] \Biggr] \\
&=\mathbb{E}_{p(x) \pi_0(a \mid x)} \Biggl[ w^{\prime}(x, a) \nabla_{\theta} \log \pi_{\theta}(a \mid x) \\
&\quad\quad \times ( S(x, a, t) - \hat{S}(x, a, t) ) \\
&\quad\quad + \mathbb{E}_{\pi_{\theta}(a \mid x)} [ (\hat{S}(x, a, t) - \lambda c(x, a)) \nabla_{\theta} \log \pi_{\theta}(a \mid x) ] \Biggr].
\end{align*}

\textbf{(i) Case when the propensity score $\pi_0(a \mid x)$ is correctly specified:}
\begin{align*}
&\mathbb{E}_{p(x) \pi_0(a \mid x)} [ w^{\prime}(x, a) \nabla_{\theta} \log \pi_{\theta}(a \mid x) ( S(x, a, t) - \hat{S}(x, a, t) ) \\
&\quad+ \mathbb{E}_{\pi_{\theta}(a \mid x)} [ (\hat{S}(x, a, t) - \lambda c(x, a)) \nabla_{\theta} \log \pi_{\theta}(a \mid x) ] ] \\
&=\mathbb{E}_{p(x)} [ \mathbb{E}_{\pi_\theta(a \mid x)} [ \nabla_{\theta} \log \pi_{\theta}(a \mid x) ( S(x, a, t) - \hat{S}(x, a, t) ) ] \\
&\quad+ \mathbb{E}_{\pi_{\theta}(a \mid x)} [ (\hat{S}(x, a, t) - \lambda c(x, a)) \nabla_{\theta} \log \pi_{\theta}(a \mid x) ] ] \\
&=\mathbb{E}_{p(x)} [ \mathbb{E}_{\pi_\theta(a \mid x)} [ \nabla_{\theta} \log \pi_{\theta}(a \mid x) ( S(x, a, t) - \hat{S}(x, a, t) )  \\
&\quad+ (\hat{S}(x, a, t) - \lambda c(x, a)) \nabla_{\theta} \log \pi_{\theta}(a \mid x) ] ] \\
&=\mathbb{E}_{p(x) \pi_\theta(a \mid x)} [ \nabla_{\theta} \log \pi_{\theta}(a \mid x) R_\lambda(x, a, t) ] \\
&= \nabla_{\theta} L(\pi_{\theta}, t, \lambda).
\end{align*}

\textbf{(ii) Case when the outcome model $\hat{S}(x, a, t)$ is correctly specified:}
\begin{align*}
&\mathbb{E}_{p(x) \pi_0(a \mid x)} [ w^{\prime}(x, a) \nabla_{\theta} \log \pi_{\theta}(a \mid x) ( S(x, a, t) - S(x, a, t) ) \\
&\quad+ \mathbb{E}_{\pi_{\theta}(a \mid x)} [ (S(x, a, t) - \lambda c(x, a)) \nabla_{\theta} \log \pi_{\theta}(a \mid x) ] ] \\
&=\mathbb{E}_{\pi_{\theta}(a \mid x)} [ R_{\lambda}(x, a, t) \nabla_{\theta} \log \pi_{\theta}(a \mid x) ] ] \\
&= \nabla_{\theta} L(\pi_{\theta}, t, \lambda).
\end{align*}
Therefore, in either case, the IPCW-DR estimator for the Lagrangian policy gradient is an unbiased estimator of the true Lagrangian policy gradient $\nabla_{\theta} L(\pi_{\theta}, t, \lambda)$.

\end{proof}

\end{document}